\documentclass[11pt]{amsart}
\usepackage{graphicx,amssymb,amsmath,amsthm}
\usepackage{enumerate}
\usepackage{dsfont}
\usepackage[colorlinks, citecolor=red]{hyperref}
\usepackage{comment,cite,color}
\usepackage{cite,color}
\usepackage{mathrsfs}
\usepackage{epsfig}
\usepackage{lscape}
\usepackage{subfigure}
\usepackage{epstopdf}

\usepackage{caption}
\usepackage{algorithm}
\usepackage{algpseudocode}
\numberwithin{equation}{section}

\newcommand{\innerp}[1]{\langle {#1} \rangle}
\newcommand{\norm}[1]{\|{#1}\|_2}

\newcommand{\norms}[1]{\|{#1}\|}
\newcommand{\abs}[1]{\lvert#1\rvert}

\newcommand{\PP}{{\mathbb P}}
\newcommand{\E}{{\mathbb E}}

\newcommand{\R}{{\mathbb R}}

\newcommand{\T}{\top}
\newcommand{\C}{{\mathbb C}}

\newcommand{\vx}{{\mathbf x}}
\newcommand{\vy}{{\mathbf y}}

\newcommand{\vz}{{\mathbf z}}

\newcommand{\ve}{{\mathbf e}}

\newcommand{\vA}{{\mathbf A}}

\newtheorem{definition}{Definition}[section]
\newtheorem{corollary}[definition]{Corollary}
\newtheorem{prop}[definition]{Proposition}
\newtheorem{theorem}[definition]{Theorem}

\newtheorem{remark}[definition]{Remark}

\date{}

\begin{document}
\baselineskip 18pt
\bibliographystyle{plain}
\title[RIP for structural random matrices]{Improved bounds for the RIP of
Subsampled Circulant matrices }

\author{Meng Huang}
\address{LSEC, Inst.~Comp.~Math., Academy of
Mathematics and System Science,  Chinese Academy of Sciences, Beijing, 100091, China}
\email{hm@lsec.cc.ac.cn}

\author{Yuxuan Pang}
\address{LSEC, Inst.~Comp.~Math., Academy of
Mathematics and System Science,  Chinese Academy of Sciences, Beijing, 100091, China}
\email{pangyuxuan@lsec.cc.ac.cn}

\author{Zhiqiang Xu}
\thanks{Zhiqiang Xu was supported  by NSFC grant (91630203, 11331012, 11688101) and by National Basic Research Program of China (973 Program 2015CB856000).}
\address{LSEC, Inst.~Comp.~Math., Academy of
Mathematics and System Science,  Chinese Academy of Sciences, Beijing, 100091, China\newline
School of Mathematical Sciences, University of Chinese Academy of Sciences, Beijing 100049, China}
\email{xuzq@lsec.cc.ac.cn}

\begin{abstract}
In this paper, we study the restricted isometry property of partial random circulant matrices. For a bounded subgaussian generator with independent entries, we prove that the partial random circulant matrices satisfy $s$-order RIP with high probability if one chooses $m\gtrsim s \log^2(s)\log (n)$ rows  randomly where  $n$ is the vector length.
This improves  the  previously known bound  $m \gtrsim s \log^2 s\log^2 n$.
\end{abstract}
\maketitle
\section{Introduction}
\subsection{Compressed sensing}
The aim of
compressed sensing \cite{candes2006robust,donoho2006compressed,foucart2013mathematical,rauhut2010compressive}
 is to recover $s$-sparse signals $\vx\in \C^n$ from the linear measurements $\vy=\vA\vx $ and the compressed sensing matrix  $\vA\in \C^{m\times n}$ with $m<n$. Here, we say $\vx\in \C^n$ is $s$-sparse if $\|\vx\|_0\leq s$ where $\|\vx\|_0$ denotes the number of nonzero entries of $\vx$.
 A naive approach for reconstructing $\vx$ is to solve the following $\ell_0$-minimization problem
\begin{equation*}
(P_0)\quad    \mathop{\min} \limits_{\vz}  \norms{\vz}_0 \qquad \text{subject to} \quad \vA\vz=\vy.
\end{equation*}
However, the $\ell_0$-minimization problem is NP-hard \cite{natarajan1995sparse} which is not tractable. A natural  approach is to relax $\ell_0$-minimization to $\ell_1$-minimization \cite{candes2006robust,chen2001atomic,donoho2006compressed,foucart2013mathematical}, i.e.,
\begin{equation*}
   (P_1)\quad  \mathop{\min} \limits_{\vz}  \norms{\vz}_1 \qquad \text{subject to} \quad \vA\vz=\vy.
\end{equation*}
The $(P_1)$ is a convex optimization problem which can be solved efficiently.
To guarantee the reconstruction of the  sparse signals $\vx$ by $(P_1)$, it suffices to require the matrix $\vA$ satisfying restricted isometry property (RIP).
For a matrix $\vA\in \R^{m\times n}$ and an integer  $s\in [1,n)$, we say $\vA$ satisfies $s$-order RIP with  constant $\delta_s\in [0,1)$ if
\begin{equation} \label{RIP}
  (1-\delta_s)\norm{\vx}^2\le \norm{\vA\vx}^2 \le (1+\delta_s)\norm{\vx}^2
\end{equation}
holds for all $s$-sparse signals $\vx$. It has been proved that if the matrix $\vA$ satisfies $t\cdot s$-order  RIP with $\delta_{ts}<\sqrt{1-\frac{1}{t}}$ for $t>4/3$, then every $s$-sparse vector $\vx$ can be  reconstructed  by solving $(P_1)$ \cite{CJT, CZRIP}. Moreover, RIP is also employed to study the performance of greedy  algorithms for the recovery of sparse signals, such as OMP \cite{zhang2011sparse}, OMMP \cite{OMMP}, CoSaMP \cite{foucart2012sparse,needell2009cosamp}, iterative hard thresholding \cite{blumensath2009iterative} and hard thresholding pursuit \cite{foucart2011hard}. Hence, one is interested in constructing RIP matrices $\vA\in \C^{m\times n}$ with $m$ being as small as possible. A popular method for constructing RIP matrices is to use random matrices.  For example, Gaussian random matrix $\frac{1}{\sqrt{m}}\vA\in \R^{m\times n}$ satisfies $s$-order RIP with high probability provided that $m\ge Cs\log (n/s)$,
where the entries of $\vA$ are independent standard normal random variables. From the  Gelfand widths, the lower  bound $Cs\log (n/s)$ is optimal up to a constant \cite{foucart2010gelfand}. Others random matrices which can achieve this  bound include Bernoulli matrices and subgaussian matrices \cite{mendelson2008uniform,rauhut2010compressive}.
 In practical applications, one prefers structure random matrices  since they  can make the recovery algorithms more efficient. An important example of structure random matrices is partial  Fourier matrices whose rows are selected randomly from the discrete Fourier matrix. It has been showed that the partial Fourier matrices satisfy $s$-order RIP  with high probability provided that
$
m\ge Cs\log^2 s \log n ,
$
(see \cite{haviv2017restricted,rauhut2010compressive,rudelson2008sparse,bourgain2014improved}).

The aim of this paper is to study the RIP of the partial random circulant matrices. Compared to Bernoulli or Gaussian matrices, the partial random circulant matrices have the advantage that they reduce the generation of only $n$ independent random variables instead of $n^2$. More importantly, they admit fast matrix-vector multiplication and arise naturally in certain applications such as  in radar, aperture imaging \cite{haupt2010toeplitz,romberg2009compressive} as well as MR imaging \cite{liang2009toeplitz}. Hence, they attract much attention.

\subsection{Related work}
Assume that $\xi=(\xi_1,\ldots,\xi_n)\in \C^n$. We define the circulant matrix generated by $\xi$ as
\begin{eqnarray*}
  \vA_{\xi} &:=& \left(
                \begin{array}{cccc}
                  \xi_1 & \xi_n & \cdots & \xi_2 \\
                  \xi_2 & \xi_1 & \cdots & \xi_3 \\
                  \vdots & \vdots & \ddots & \vdots \\
                  \xi_n & \xi_{n-1} & \cdots & \xi_1 \\
                \end{array}
              \right)\in \C^{n\times n}.
\end{eqnarray*}
For a multiset $\Omega$ in $\{1,\ldots,n\}$ with cardinality $m$, let $P_\Omega:\C^n\rightarrow \C^m$ denote the projection operator that restricts a vector $\vx\in \C^n$ to its entries in $\Omega$. Then the corresponding partial circulant matrix generated by  $\xi$ is defined as
\begin{equation*}
  \Phi_\xi=\frac{1}{\sqrt{m}}P_\Omega \vA_\xi \in \C^{m\times n}.
\end{equation*}

We next introduce another structured random matrices. For $\xi=(\xi_1,\ldots,\xi_{2n-1})\in \C^{2n-1}$, the Toeplitz matrix $T_\xi$ is defined as
\begin{equation*}
  T_{\xi}=\left(
                \begin{array}{cccc}
                  \xi_n & \xi_{n-1} & \cdots & \xi_1 \\
                  \xi_{n+1} & \xi_n & \cdots & \xi_{2} \\
                  \vdots & \vdots & \ddots & \vdots \\
                  \xi_{2n-1} & \xi_{2n-2} & \cdots & \xi_{n} \\
                \end{array}
              \right)
\end{equation*}
and the Hankel matrix is defined as  $H_{\xi}=T_{\xi}J_n$  where $J_n=[\ve_n,\ldots,\ve_1]$ and $\ve_j,j=1,\ldots,n$ are the standard orthogonal vectors.
For a multiset $\Omega$ of $\{1,\ldots,n\}$ with cardinality $m$, the corresponding partial Toeplitz matrix
\begin{equation} \label{hankel matrix}
  \Psi_\xi=\frac{1}{\sqrt{m}}P_\Omega T_\xi \in \C^{m\times n}.
\end{equation}
Similarly, we can define the partial Hankel matrix.
  A simple observation is that
$\frac{1}{\sqrt{m}}P_\Omega T_\xi$ satisfies $s$-order RIP with the constant $\delta_s$ if and only if $\frac{1}{\sqrt{m}}P_\Omega H_\xi$ satisfies the same property.
As we will see later, a Toeplitz matrix  can be embedded in a circulant matrix of twice the dimension. Hence, many RIP results for partial random circulant matrices can be extended to the partial random Toeplitz matrices as well as partial random Hankel matrices.

 In \cite{romberg2009compressive}, under the setup
 in which both $\Omega$ as well as the generating vector $\xi$ are chosen at random,
 Romberg proves that $m\gtrsim s\log^6 n$ measurements are sufficient to guarantee  $\Phi_\xi=\frac{1}{\sqrt{m}}P_\Omega \vA_\xi \in \C^{m\times n}$ satisfying RIP.   For an arbitrary fixed selection of  $\Omega\subset \{1,\ldots,n\}$, the first theoretical results are established in \cite{bajwa2007toeplitz} with showing that $m\times n$ partial random circulant matrices satisfy $s$-order RIP  with high probability provided $m\gtrsim s^3\log n$. This is then improved by Bajwa et al. \cite{bajwa2008compressed,haupt2010toeplitz} to $m \gtrsim s^2\log n$ for general matrices whose entries were drawn from bounded or Gaussian distribution. Later, Rauhut, Romberg and Tropp prove that $m\gtrsim s^{3/2} \log^{3/2} n$ measurements are enough. And the result is then improved by Krahmer, Mendelson and Rauhut \cite{krahmer2014suprema} to $O(s\log^2 s \log^2 n)$. Recently, Mendelson, Rauhut and Ward \cite{mendelson2016improved} prove that $O(s \log^2 s \log(\log s) \log n)$ measurements for partial random circulant matrix with random sampling set $\Omega$ and Gaussian random generator are sufficient to recover all $s$-sparse vectors with high probability via $\ell_1$-minimization, but it does not establish any type of RIP.

\subsection{Main results}
In this paper, we study the RIP of the  partial random circulant matrices  with bounded or Gaussian entries, where the subsampling locations  $\Omega$ are selected at random from $\{1,\ldots,n\}$. We show that, if $m\gtrsim s\log^2 s \log n $ then   $\Phi_\xi= \frac{1}{\sqrt{m}}P_\Omega \vA_\xi$
satisfies $s$-order RIP with high probability where $\xi$ is a bounded random vector.
 This improves upon the best previously known bound  $O(s \log^2 s\log^2 n)$ \cite{krahmer2014suprema}.
 Some suitable bounded random vectors $\xi$ include uniform distribution $\xi_i\sim \mathrm{U}(-\sqrt{3},\sqrt{3})$ and Rademacher vector with $\mathbb{P}(\xi_i=\pm 1)=1/2$.

\begin{theorem}\label{bound circulant}
Let $\xi=(\xi_1,\ldots,\xi_n)\in \C^n$ be a random vector whose entries are i.i.d. realizations of bounded zero-mean random variables satisfying $\E \xi_j^2=1$ and  $\abs\xi_j\le c$ for some $c\ge 1$.
Suppose that $\delta>0$ is a sufficient small constant
 and  $m \gtrsim \log^2 (1/\delta)\cdot\delta^{-2}s\log^2(s/\delta)\log n$. Let $\Omega$ be a multiset of $m$ uniform and independent random elements of $\{1,\ldots,n\}$ and $\Phi_\xi= \frac{1}{\sqrt{m}}P_\Omega \vA_\xi  \in \C^{m\times n}$ be a partial random circulant matrix generated by $\xi$ and  $\Omega$. If $s \lesssim n/(\log^4 n)$, then the matrix $\Phi_\xi$  satisfies the restricted isometry property with order $s$ and constant $\delta_s\le \delta$
with probability at least $1-2^{-C\log n\log(s/\delta)}-n^{-\log n \log^2 s}$.
\end{theorem}

 In the practical applications, one usually requires that  $s<<n$.
Hence, the assumption of  $s \lesssim n/\log^4 n$ in Theorem \ref{bound circulant} is enough for the practical applications. The idea for proving Theorem \ref{bound circulant} is to use result from \cite{haviv2017restricted} (see Theorem \ref{lemma2} in Section 2), which shows the connection between $\frac{1}{n}\|M \vx\|_2^2$ and $\frac{1}{m}\|P_\Omega M \vx\|_2^2$ for any matrix $M\in \C^{n\times n}$ (see (\ref{eq:deng}) for detail).
In \cite{haviv2017restricted}, Haviv and Regev  consider the case where $M$ is a DFT matrix for which $\frac{1}{n}\|M \vx\|_2^2=\|\vx\|_2^2$.
A particular case of Theorem \ref{lemma1} from \cite{krahmer2014suprema} shows that the circulant matrix $\vA_\xi$ satisfies $\frac{1}{n}\|\vA_\xi\vx\|^2\approx \|\vx\|^2$ provided $s\lesssim n/\log^4n$ (see Corollary \ref{co:Axi}).
Combining them, we can arrive at Theorem \ref{bound circulant}.
The proof just combines the known results, but Theorem \ref{bound circulant} definitely goes  far beyond the current state-of-the-art.

\begin{remark}
We  consider the case with removing the assumption of $\xi_j$ being a bounded random variable.
Assume that $\xi \in \C^n$ is a random vector with independent, mean $0$ and variance $1$, $L$-subgaussian entries.
We can obtain that  $\PP\left(\abs{\xi_j}> L\sqrt{2\log n}\right)\le 2/n^2$ for any $\xi_j$. Then
\[
\PP\left(\max\limits_{j}\abs{\xi_j}\leq  L\sqrt{2\log n}\right)\geq 1- 2/n,
\]
which implies that $\abs{\xi_j} \leq L\sqrt{2\log n}$ for all $j=1,\ldots,n$ with probability at least $1-2/n$.
 Since the bound $L\sqrt{2\log n}$ is not a constant,
 it leads to slightly larger samples $m\gtrsim s\log^2(s\log n)\log^2 n$ for the partial random circulant matrix $\Phi_\xi$ to satisfy $s$-order RIP. Nevertheless, the result can match the best bound  $m\gtrsim s\log^2( s) \log^2 (n)$ which is obtained  in \cite{krahmer2014suprema}  provided $s\geq \log n$.
\end{remark}

\begin{remark}
 Note that a Toeplitz matrix can be embedded in a circulant matrix of twice the dimension. Then we can  show that Toeplitz matrix satisfies $\frac{1}{n}\|T_\xi\vx\|^2\approx \|\vx\|^2$ provided $s\lesssim n/\log^4n$.  Hence, our result can be extended to the partial random Toeplitz matrices.
\end{remark}
\section{Preliminaries}

We first introduce the definition of $L$-subgaussian random vectors which include the Rademacher vectors as well as the standard Gaussian vectors as special cases.
 To state conveniently, we use  $\mathcal{S}^{n-1}$ to denote the unit sphere in $\C^n$.
\begin{definition}[$L$-subgaussian]
A mean $0$ random vector $X\in\C^n$ is called isotropic if for every $\theta\in \mathcal{S}^{n-1}$, $\E\abs{\innerp{X,\theta}}=1$. A random vector $X$ is called $L$-subgaussian if it is isotropic and $\PP(\abs{\innerp{X,\theta}}\ge t)\le 2\exp(-t^2/2L^2)$ for every $\theta \in \mathcal{S}^{n-1}$ and any $t\ge 0$.
\end{definition}
 We introduce some known results which are useful in our analysis.

\begin{theorem}  (\cite[Theorem 4.1]{krahmer2014suprema}) \label{lemma1}
For any fixed subset $\Omega\subset \{1,\ldots,n\}$ and a random vector $\xi \in \C^n$ with independent, mean $0$ and variance $1$, $L$-subgaussian entries. Let $\Phi_\xi=\frac{1}{\sqrt{m}}P_{\Omega} \vA_\xi\in\R^{m\times n}$ be a partial random circulant matrix generated by $\xi$ and $\Omega$. If
\begin{equation*}
  m\ge c\delta^{-2}s \log^2 s\log^2 n,
\end{equation*}
then with probability at least $1-n^{-\log n \log^2 s}$, the matrix $\Phi_\xi$ satisfies the restricted isometry property with constant $\delta_s\le \delta$. The constant $c>0$ is universal.
\end{theorem}
If we take $\Omega=\{1,\ldots,n\}$, then we have  the following corollary:
\begin{corollary}\label{co:Axi}
Suppose that $\xi\in \C^n$ is a random vector  with independent, mean $0$ and variance $1$, $L$-subgaussian entries. Suppose that $s\lesssim n/\log^4 n$. Then
\[
(1-\delta)\|\vx\|^2\leq \frac{1}{n}\|\vA_\xi\vx\|^2\leq (1+\delta)\|\vx\|^2\quad \text{ for all $s$-sparse $\vx\in \C^n$ }
\]
holds  with probability at least $1-n^{-\log n \log^2 s}$.
\end{corollary}
 The above corollary shows that we can obtain $(1-\delta)\|\vx\|^2\leq \frac{1}{{n}}\|\vA_\xi \vx\|^2\leq (1+\delta)\|\vx\|^2$ for all $s$-sparse vector $\vx$ provided $s\lesssim n/\log^4n$.

We next introduce the main result in \cite{haviv2017restricted} which also plays an important role in our analysis.
\begin{theorem} (\cite[Theorem 4.1]{haviv2017restricted}) \label{lemma2}
For a sufficiently large $n$, a matrix $M\in \C^{n\times n} $, and sufficiently small $\varepsilon,\eta>0$, the following holds. For some $m=O\left(\log^2 (1/\varepsilon)\cdot\varepsilon^{-1}\eta^{-1}\log n \cdot\log^2(1/\eta)\right)$, let $\Omega$ be a multiset of $m$ uniform and independent random elements of $\{1,\ldots,n\}$. Then, with probability $1-2^{-\Omega(\log n \cdot\log(1/\eta))}$, it holds that for every $\vx\in \C^n$,
\begin{equation}\label{eq:deng}
  \frac{1}{n}(1-\varepsilon)\norm{M\vx}^2-\eta\norms{\vx}_1^2\norms{M}_\infty^2 \le\frac{1}{m}\sum_{j\in \Omega } |(M\vx)_j|^2\le \frac{1}{n}(1+\varepsilon)\norm{M\vx}^2+\eta\norms{\vx}_1^2\norms{M}_\infty^2,
\end{equation}
where  $\norms{M}_\infty:=\max_{i,j}\abs{M_{i,j}}$.
\end{theorem}

\section{Proofs of Theorem \ref{bound circulant}}
Before giving the proof of Theorem \ref{bound circulant}, we next introduce a proposition which shows that $\frac{1}{n}\|\vA_\xi \vx\|^2 \thickapprox \|\vx\|^2$ does not hold for all $\vx\in \C^n$ with high probability. Hence, to guarantee $\frac{1}{n}\|\vA_\xi\vx\|^2 \thickapprox \|\vx\|^2$ we need to require $\vx$ lies in some subset in $\C^n$. For example, in Theorem \ref{bound circulant}, we require $s=\|\vx\|_0\lesssim n/\log^4 n$. This also shows the essential difference between $\vA_\xi$ and Fourier matrices.

To state conveniently, we say a vector is a Gaussian random vector if the entries are i.i.d. standard Gaussian random variables.
\begin{prop}\label{pr:deng}
Let $\xi=(\xi_1,\ldots,\xi_n)\in \R^n$ be a Rademacher vector or Gaussian random vector and $\vA_\xi \in \R^{n\times n}$ be the random circulant matrix generated by $\xi$.
Then for any fixed $\epsilon>0$, there exists a vector $\vx \in \R^n$ and a positive constant $p_0$ so that
\begin{equation*}
  \PP\left(\frac{1}{n}\norm{\vA_\xi \vx}^2 < \epsilon \right) \ge p_0 \quad  \text{when $n$ is large enough}.
\end{equation*}
\end{prop}
\begin{proof}
Setting $\vx =\frac{1}{\sqrt{n}}(1,\ldots,1)^\T$, we have
\begin{equation*}
  \frac{1}{n}\norm{\vA_\xi \vx}^2= \left(\frac{\xi_1+\cdots+\xi_n}{\sqrt{n}}\right)^2:=\zeta_n^2.
\end{equation*}
For Gaussian random vector $\xi$, a simple observation is that  $\zeta_n \sim \mathcal{N}(0,1)$. It implies that
\begin{equation*}
  \PP\left(\frac{1}{n}\norm{\vA_\xi \vx}^2 < \epsilon \right) = \PP(\abs{\zeta_n} < \sqrt{\epsilon})= 2(1-\Phi(\sqrt{\epsilon})):=p_0,
\end{equation*}
where $\Phi(x)$ is the cumulative distribution function of standard Gaussian random variable. For Rademacher vector $\xi$, note that $\{\xi_i\}$ is a i.i.d. random variable sequence with  $\E(\xi_i)=0$ and  $\E(\xi_i^2)=1$.
Recall that
\[
\PP\left(\frac{1}{n}\norm{\vA_\xi \vx}^2 < \epsilon \right) = \PP(\abs{\zeta_n} < \sqrt{\epsilon}).
\]
Then from Central Limit Theorem, we obtain that $\PP(\abs{\zeta_n} < \sqrt{\epsilon})$ tends to $2(1-\Phi(\sqrt{\epsilon}))$ with $n$ tending to infinity. Hence, we arrive at conclusion.
\end{proof}

We next give the proof of Theorem \ref{bound circulant}.

\begin{proof}[Proof of Theorem \ref{bound circulant}] Through out this proof, we assume that $\|\vx\|_0\leq s$.
Recall that $\Phi_\xi=\frac{1}{\sqrt{m}}P_\Omega \vA_\xi$. From Corollary \ref{co:Axi}, we obtain that
\begin{equation} \label{circulant}
 (1-\delta/4)\norm{\vx}^2 \le \frac{1}{n}\norm{\vA_\xi \vx}^2\le (1+\delta/4)\norm{\vx}^2, \quad \text{ for all  $s$ -sparse vectors } \vx
\end{equation}
holds
with probability at least $1-n^{-\log n \log^2 s}$ provided $s\lesssim \delta^2n/(\log^2 s\log^2 n)$.
 Cauchy-Schwarz inequality implies that  $\norms{\vx}_1\le \sqrt{s}\norm{\vx}$. Observe that
 \[
\norm{\Phi_\xi \vx}^2=\frac{1}{m}\sum_{j\in \Omega} |(\vA_\xi \vx)_j|^2.
\]
For the given multiset $\Omega$, choosing $\eta\ge \delta/(4c^2s) $ in Theorem \ref{lemma2} for a fixed circulant matrix $\vA_\xi$ with
$\norms{\vA_\xi}_\infty\le c$, then with probability at least $1-2^{-C\log n\log(s/\delta)}$ we have
\begin{equation}\label{submatrix}
  \frac{1}{n}(1-\delta/4)\norm{\vA_\xi \vx}^2-\delta/4\norms{\vx}_2^2 \le\norm{\Phi_\xi \vx}^2\le
  \frac{1}{n}(1+\delta/4)\norm{\vA_\xi \vx}^2+\delta/4\norms{\vx}_2^2
\end{equation}
provided $m=O\left(\log^2 (1/\delta)\cdot\delta^{-2}s\log^2(s/\delta)\log n\right)$.

Combining (\ref{circulant}) and (\ref{submatrix}), we obtain that
\begin{equation*}
 (1-\delta)\norm{\vx}^2\le \norm{\Phi_\xi \vx}^2\le (1+\delta)\norm{\vx}^2
\end{equation*}
holds with probability at least $1-2^{-C\log n\log(s/\delta)}-n^{-\log n \log^2 s}$, which arrives at the conclusion.

\end{proof}

\end{document}